\documentclass[10pt]{article}
\usepackage{verbatim,color,amssymb,epsfig}
\usepackage{amsmath,amsthm,graphicx,bbm,url}
\usepackage[usenames,dvipsnames]{xcolor}
\usepackage{fancyhdr}
\usepackage[authoryear, sort]{natbib}
\usepackage{subfig}
\usepackage{enumitem}
\usepackage{float}
\usepackage{pgfplots}

\newcommand{\E}{\mathbb{E}}

\newcommand{\argmin}[1]{\underset{#1}{\operatorname{argmin}}\text{ }}

\newcommand{\ind}[1]{\mathbbm{1}_{#1}}

\newtheorem{Thm}{\underline{\bf Theorem}}
\newtheorem{Assume}{\underline{\bf Assumptions}}

\newtheorem{Lem}{\underline{\bf Lemma}}

\begin{document}

\thispagestyle{empty}
\baselineskip=28pt
\begin{center}
  {\LARGE{\bf A Flexible Procedure for Mixture Proportion Estimation in Positive--Unlabeled Learning}}
 % {\LARGE{\bf Mixture Proportion Estimation for Positive--Unlabeled Learning via Classifier Dimension Reduction}}
\end{center}

\baselineskip=12pt

\vskip 2mm
\begin{center}
  %\hskip 5mm\\ \hskip 5mm\\
  Zhenfeng Lin\\
  Department of Statistics, Texas A\&M University\\
  3143 TAMU, College Station, TX 77843-3143\\
  zflin@stat.tamu.edu\\
  \hskip 5mm \\
  James P. Long\\
  Department of Biostatistics, University of Texas MD Anderson Cancer Center\\
  P.O. Box 301402, Houston, TX 77230-1402\\
  jlong@stat.tamu.edu\\
\end{center}

\hskip 5mm

\begin{center}
  {\Large{\bf Abstract}}
\end{center}
\baselineskip=18pt

%\address[3]{\orgdiv{Org Division}, \orgname{Org Name}, \orgaddress{\state{State name}, \country{Country name}}}

Positive--unlabeled (PU) learning considers two samples, a positive set $P$ with observations from only one class and an unlabeled set $U$ with observations from two classes. The goal is to classify observations in $U$. Class mixture proportion estimation (MPE) in $U$ is a key step in PU learning.  \cite{Blanchard2010} showed that MPE in PU learning is a generalization of the problem of estimating the proportion of true null hypotheses in multiple testing problems. Motivated by this idea, we propose reducing the problem to one dimension via construction of a probabilistic classifier trained on the $P$ and $U$ data sets followed by application of a one--dimensional mixture proportion method from the multiple testing literature to the observation class probabilities. The flexibility of this framework lies in the freedom to choose the classifier and the one--dimensional MPE method. We prove consistency of two mixture proportion estimators using bounds from empirical process theory, develop tuning parameter free implementations, and demonstrate that they have competitive performance on simulated waveform data and a protein signaling problem.

\baselineskip=12pt
\par\vfill\noindent

\par\medskip\noindent
\underline{\bf Short title}: Mixture Proportion Estimation for PU Learning

\clearpage\pagebreak\newpage
\pagenumbering{arabic}
\newlength{\gnat}
\setlength{\gnat}{22pt}
\baselineskip=\gnat

\section{Introduction}\label{sec:intro}

Let
\begin{align}
  \label{eq:model1}
  &X_{1},\ldots,X_{n} \sim F = \alpha F_0 + (1-\alpha)F_1, \\
  &X_{L,1},\ldots,X_{L,m} \sim F_1, 
   \nonumber
\end{align}
all independent, where $F_0$ and $F_1$ are distributions on $\mathbb{R}^p$ with densities $f_0$ and $f_1$ with respect to measure $\mu$. The goal is to estimate $\alpha$ and the classifier
\begin{equation} \label{eq:true_classifier}
C_{01}(x) = \frac{(1-\alpha) f_1(x)}{\alpha f_0(x) + (1-\alpha) f_1(x)},
\end{equation}
which can be used to separate the unlabeled data $\{X_{i}\}_{i=1}^n$ into the classes $0$ and $1$. The above problem has been termed \textit{Learning from Positive and Unlabeled Examples}, \textit{Presence Only Data}, \textit{Partially Supervised Classification}, and the \textit{Noisy Label Problem} in the machine learning literature \citep{elkan2008learning,ward2009presence,Ramaswamy2016,scott2013classification,scott2015rate,liu2002partially}.  In this work, we use the term PU learning to refer to Model \eqref{eq:model1}. Here we denote the positive set $P := \{X_{L,i}\}_{i=1}^m$ and the unlabeled set $U := \{X_i\}_{i=1}^n$. This setting is more challenging than the traditional classification framework where one possesses labeled training data belonging to both classes. In particular $\alpha$ and $C_{01}$ are not generally identifiable from the $P$ and $U$ data. PU learning has been applied to text analysis \citep{liu2002partially}, time series \citep{Nguyen2011}, bioinformatics \citep{YangPeng2012}, ecology \citep{ward2009presence}, and social networks \citep{Chang2016}.  

%This setting is more challenging than the traditional classification framework where one possesses labeled training data belonging to both classes. In particular $\alpha$ and $C_{01}$ are not generally identifiable from the data $\{X_i\}_{i=1}^{n}$ and $\{X_{L,i}\}_{i=1}^{m}$. Various works have suggested methods for addressing identifiability and estimating the classifier. These include assuming $\alpha$ is known \citep{ward2009presence} and making assumptions on the support of $f_0$ and $f_1$ \citep{scott2013classification,scott2015rate}.

Several strategies have been proposed for solving the PU problem. \cite{ward2009presence} assumes $\alpha$ is known and uses logistic regression to classify $U$. The SPY method of \cite{liu2002partially} classifies $U$ directly by identifying a ``reliable negative set.'' The SPY method has practical challenges including choosing the reliable negative set. Other strategies estimate $\alpha$ directly. \cite{Ramaswamy2016} estimate $\alpha$ via kernel embedding of distributions. \cite{scott2015rate} and \cite{Blanchard2010} estimate $\alpha$ using the ROC curve produced by a classifier trained on $P$ and $U$.

%Estimating the proportion of true null hypotheses in multiple testing problems is an important step for maximizing power while controlling the False Discovery Rate (FDR).

\cite{Blanchard2010} showed that MPE in the PU model is a generalization of estimating the proportion of true nulls in multiple testing problems. Specifically, suppose that $F_0$ and $F_1$ are one--dimensional distributions and $F_1$ is known. Then the unlabeled set $X_{1},\ldots,X_{n}$ may be interpreted as test statistics with the hypotheses:
\begin{align*}
  &H_0: X_{i} \sim F_1,\\
  &H_a: X_{i} \sim F_0.
\end{align*}
In this context, $1-\alpha$ is the proportion of true null hypotheses and the classifier $C_{01}$ is the local FDR \citep{efron2001empirical}. There are many works on addressing identifiability and estimation of $\alpha$ and $C_{01}$ in this simpler setting \citep{patra2015estimation,efron2012large,genovese2004stochastic,robin2007semi,meinshausen2006estimating}.

FDR $\alpha$ estimation methods have been developed for one--dimensional MPE problems and are not directly applicable on the multidimensional PU learning problem in which $X_i \in \mathbb{R}^p$. In this work, we show that the PU MPE problem can be reduced to dimension one by constructing a classifier on the P versus U data sets followed by transforming observations to class probabilities. One dimensional MPE methods from the FDR literature can then be applied to the class probabilities. Computer implementation of this approach is straightforward because one can use existing classifier and one--dimensional MPE algorithms. We prove consistency for adaptations of two one--dimensional MPE methods: \cite{Storey2002} based on empirical processes and \cite{patra2015estimation} based on isotonic regression. These proofs use results from empirical process theory. We show that the ROC method used in \cite{Blanchard2010} and \cite{scott2015rate} in the machine learning literature is a variant of the method proposed by Storey \cite{Storey2002} in the multiple testing literature. These results strengthen connections between the PU learning and multiple testing communities.

The rest of the paper is organized as follows. In Section \ref{sec:procedure} we give a sketch of the proposed procedure, which includes two proposed estimators C-PS and C-ROC. This section consists of three parts. First, a motivation of the procedure from the hypothesis testing perspective is explained. Second, identifiability of $\alpha$ is addressed. Third, a workflow is provided to explain how to implement the proposed procedure. In Section \ref{sec:identify} we show that Model \eqref{eq:model1} can be reduced to one-dimension with a classifier. In Section \ref{sec:estimation_alpha0} we show consistency of two $\alpha$ estimators. In Section \ref{sec:experiments} we numerically show that the estimators perform well in various settings. A conclusion is made in Section \ref{sec:conclusion}. Appendix \ref{sec:technical_notes} gives proofs of theorems in the paper. Supporting lemmas can be found in Appendix \ref{sec:lemmas}.

\section{Background and Proposed Procedure}\label{sec:procedure}

\subsection{Multiple Testing, FDR, and Estimating the Proportion of True Nulls}

Suppose one conducts $n$ tests of null hypothesis $H_{0}: X_i \sim F_1$ versus alternative hypothesis $H_{a}: X_i \sim F_0$, $i=1,\ldots, n$. The $X_i$ are typically test statistics or p--values and the null distribution $F_1$ is assumed known (usually $Unif[0,1]$ in the case of $X_i$ being p--values). The distribution of the $X_i$ are $F = \alpha F_0 + (1-\alpha) F_1$, where $1-\alpha$ is the proportion of true null hypotheses. The \textit{false discovery rate} (FDR) is the expected proportion of false rejections. If $R$ is the number of rejections and $V$ is the number of false rejections then $FDR \equiv \E[\frac{V}{R}\mathbf{1}_{R > 0}]$. \cite{BenjaminiHochberg95} developed a linear step--up procedure which bounds the FDR at a user specified level $\beta$. In fact, this procedure is conservative and results in an FDR $\leq \beta (1-\alpha) \leq \beta$. This conservative nature causes the procedure to have less power than other methods which control FDR at $\beta$. \textit{Adaptive} FDR control procedures first estimate $1-\alpha$ and then use this estimate to select a $\beta$ which ensures control at some specified level while maximizing power. Many estimators of $\alpha$ have been proposed \citep{patra2015estimation, Storey2002, Benjamini2006, langaas2005, blanchard2009adaptive, benjamini2000adaptive}.

There are two reasons why these procedures cannot be directly applied to the PU learning problem. First, many of the methods have no clear generalization to dimension greater than one because they require an ordering of the test statistics or p--values. Second, the distribution $F_1$ is assumed known where as in the PU learning problem we only have a sample from this distribution. The classifier dimension reduction procedure we outline in Section \ref{sec:workflow} addresses the first point by transforming the PU learning problem to 1--dimension. The theory we develop in Sections \ref{sec:identify} and \ref{sec:estimation_alpha0} addresses the second issue.

%The insight here is the observation that PU problem (Model \eqref{eq:model1}) can be reduced to an one-dimensional multiple testing problem. Hence, if there is a method that can reduce the dimension of $X_i$ to one, then we can directly apply those estimator for $\alpha$ in multiple testing problem. In next subsection, we address the identifiability issue in PU problem. And in the subsequent subsection, we propose a procedure that solves PU problem via a combination of dimension reduction and the estimation of the lower bound of $\alpha$. 
 
\subsection{Identifiability of $\alpha$ and $C_{01}$}

Many works in both the PU learning and multiple testing literature have discussed the non--identifiability of the parameters $\alpha$ and $F_0$. For any given $(\alpha,F_0)$ pair with $\alpha < 1$, one can find a $\gamma > 0$ such that $\alpha' \equiv \alpha + \gamma \leq 1$. Define $F_0' \equiv \frac{\alpha F_0 + \gamma F_1}{\alpha + \gamma}$. Then
\begin{equation*}
F = \alpha' F_0' + (1-\alpha')F_1,
\end{equation*}
which implies $(\alpha',F_0')$ and $(\alpha,F_0)$ result in the same distributions for $P$ and $U$.

%Additional model assumptions can solve the identifiability issue. For example, \cite{ward2009presence} assume $\alpha$ is known. In the multiple testing literature, $F_1$ may be assumed uniform on the unit interval because $X_i$ are p--values and $f_0$ decreasing \todo{cite someone}. 

To address this issue, we follow the approach taken by \cite{Blanchard2010} and \cite{patra2015estimation} and estimate a lower bound on $\alpha$ defined as
\begin{align} \label{def:alpha0}
\alpha_0 := \inf\left\lbrace \gamma \in (0,1]: \frac{F-(1-\gamma)F_1}{\gamma} \text{ is a c.d.f.}\right\rbrace.
\end{align}

The parameter $\alpha_0$ is identifiable. Recall the objective is to estimate
\begin{equation*}
C_{01}(x) = \frac{(1-\alpha) f_1(x)}{\alpha f_0(x) + (1-\alpha) f_1(x)}. 
\end{equation*}
Let $\pi=m/(m+n)$ be the proportion of labeled data. The classifier
\begin{equation*}
  C(x) = \frac{\pi f_1(x)}{\pi f_1(x) + (1-\pi)f(x)}
\end{equation*}
outputs the probability an observation is from the labeled data set at a given $x$. We can approximate $C$ by training a model on the $P$ versus $U$ data sets. The classifiers $C$ and $C_{01}$ are related through $\alpha$. To see this, note that after some algebra
\begin{equation*}
  \frac{f_1(x)}{f(x)} = \frac{C(x)}{1-C(x)} \frac{1-\pi}{\pi}.
\end{equation*}
Thus
\begin{equation*}
  C_{01}(x) = \frac{(1-\alpha) f_1(x)}{f(x)} = \frac{1-\pi}{\pi}\frac{C(x)}{1-C(x)}(1-\alpha).
\end{equation*}
Since $\alpha$ is not generally identifiable, neither is $C_{01}$. However the plug-in estimator, using $C_n$ (a classifier trained on $P$ versus $U$) and $\widehat{\alpha}_0$ (some estimator of $\alpha_0$),
\begin{equation*}
  \widehat{C}_{01}^0(x) \equiv \frac{1-\pi}{\pi}\frac{C_n(x)}{1-C_n(x)}(1-\widehat{\alpha}_0)
\end{equation*}
is an estimated upper bound for $C_{01}$. We can classify an unlabeled observation $X_i$ as being from $F_1$ if $\widehat{C}^0_{01}(X_{i}) > \frac{1}{2}$. The problem has now been reduced to estimation of $\alpha_0$. The classifier $C_n$ plays an important role in estimation of $\alpha_0$ as well, as shown in the following section.
%\todo{discuss that $C_n$ is necessary anyway to estimate $C$}

%% The parameter $\alpha_0$ is a lower bound on $\alpha$. Further using $\alpha_0$, we can construct a classifier that gives upper bound of probability of observation belonging to $F_1$ because
%% \begin{equation*}
%% C_{01}(x) = \frac{(1-\alpha) f_1(x)}{\alpha f_0(x) + (1-\alpha) f_1(x)} = \frac{f_1(x)}{\frac{\alpha}{1-\alpha} f_0(x) + f_1(x)} \leq \frac{f_1(x)}{\frac{\alpha_0}{1-\alpha_0} f_0(x) + f_1(x)} = \frac{(1-\alpha_0) f_1(x)}{\alpha_0 f_0(x) + (1-\alpha_0) f_1(x)}.
%% \end{equation*}

\subsection{Workflow for $\alpha_0$ Estimation}
\label{sec:workflow}

\begin{figure}[H]
  \begin{center}
%      \begin{includegraphics}[height=10cm,width=14cm]{figs/workflow_fig}
    \begin{includegraphics}[width=\textwidth]{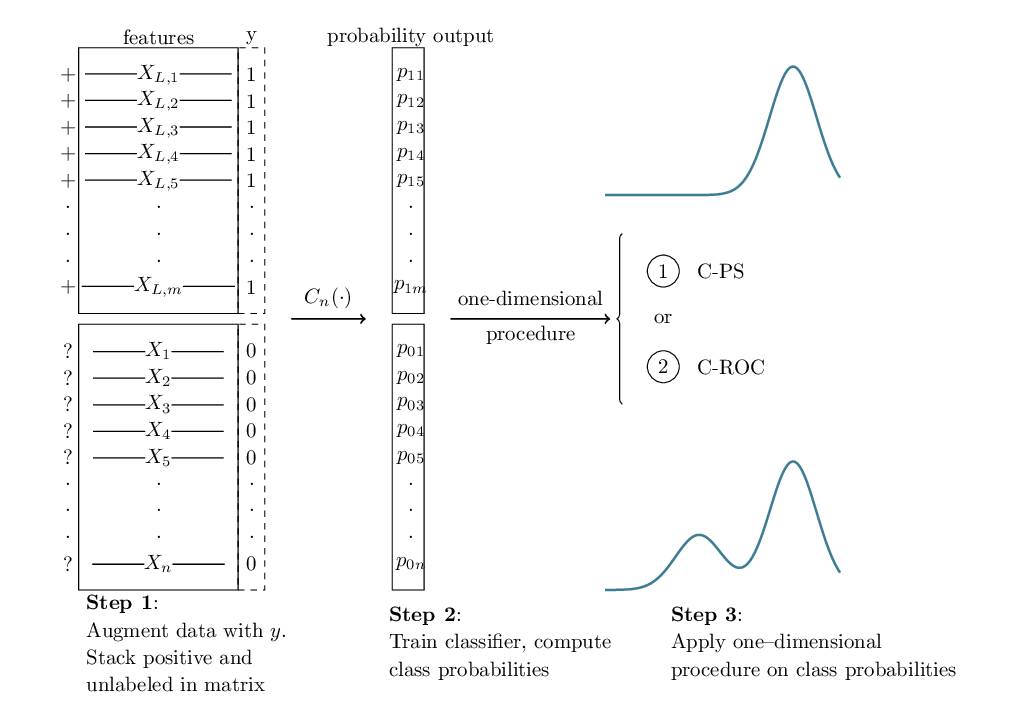}
      \caption{Workflow of proposed procedure. In \textbf{Step 1}, ``+'' denotes the positive samples, and ``?'' denotes the unlabeled samples of unknown class (can be ``+'' or ``-''). Stack the set $P$ and the set $U$ together as a large matrix, and add a new column $y$ to manually impose pseudo labels on observations: ``1'' for $X_{L,i}$ and ``0'' for $X_i$. In \textbf{Step 2}, a classifier $C_n(\cdot)$ is trained on the stacked matrix and the probability predictions ($y=1$ as reference) are obtained. In \textbf{Step 3}, a one-dimensional procedure is applied to the probability output from Step 2. In this paper, two methods C-PS and C-ROC are proposed. The upper density curve is used to demonstrate that the $\mathbf{p}_1:=\{p_{1i}\}_{i=1}^{m}$ are from one population, while the bottom density curve shows that $\mathbf{p}_0:=\{p_{0i}\}_{i=1}^{n}$ are from mixture of two populations.
      \label{fig:workflow}}
    \end{includegraphics}
  \end{center}
\end{figure}

The proposed procedure to estimate $\alpha_0$ in Model \eqref{eq:model1} is summarized in Figure \ref{fig:workflow}. The key idea of this procedure is to reduce the dimension of PU learning problem via the classifier $C_n$ trained on $P$ versus $U$ and then apply a one-dimensional MPE method on the transformed data to estimate $\alpha_0$. The procedure consists of three steps:

\begin{itemize}
\item \textbf{Step 1}. Label the $P$ samples with pseudo label ($Y=1$) and label the $U$ samples with pseudo label ($Y=0$). Hence we have $\tilde{P} := \{(X_{L,i},Y_i=1), i=1,\ldots, m\}$ and $\tilde{U} := \{(X_{i},Y_i=0), i=1,\ldots, n\}$.
\vspace{-0.2cm}
\item \textbf{Step 2}. Train a probabilistic classifier $C_n(\cdot) = \widehat{P}(Y=1|X=\cdot)$ on $\tilde{P}$ versus $\tilde{U}$. Compute probabilistic predictions: $\mathbf{p}_1 := \{p_{1i}, i=1,\ldots,m\}$ and $\mathbf{p}_0 :=\{p_{0i}, i=1,\ldots,n\}$, where $p_{1i} := C_n(X_{L,i})$ and $p_{0i} := C_n(X_{i})$.
\vspace{-0.2cm}
\item \textbf{Step 3}. Apply a one-dimensional MPE method to $\mathbf{p}_1$ and $\mathbf{p}_0$ to estimate $\alpha_0$.
\end{itemize}

We augment the original data with pseudo labels in Step 1, in order to use a supervised learning classification algorithm. In Step 2 we use Random Forest \citep{Breiman2001}. However in principle any classifier can be used. Note that the $p_{0i}$ and $p_{1i}$ are scalars. Hence in Step 3 we can utilize any one-dimensional method to estimate $\alpha_0$. In this work we adapt two methods -- one from \cite{Storey2002} and \cite{scott2015rate}, another from \cite{patra2015estimation}. Note that the original theory developed for these methods assumed that the null distribution is known, but in the PU problem we need to estimate it from $\mathbf{p_1}$. Since this setting is more complex, new theory is needed. In Section \ref{sec:estimation_alpha0}, we prove the consistency of two estimators in the PU setting, using Theorems \ref{thm:change_problems} and \ref{thm:CD_donsker}.

\section{Dimension Reduction via Classifier} \label{sec:identify}

%% This section discusses the details of Step 2 in Figure \ref{fig:workflow}.

%% Methods by \cite{Storey2002} and \cite{patra2015estimation} can be used to estimate $\alpha_0$ in Equation \eqref{def:alpha0} but are only directly applicable if $F$ and $F_1$ are both one-dimensional. In order to apply these one-dimensional methods, we can firstly reduce the data in Model \eqref{eq:model1} to one dimension.

Using the $P$ and $U$ samples we can make probabilistic predictions, i.e. compute the probability that the observation is from distribution $F_1$ versus from distribution $F$. The true classifier is
\begin{equation*}
C(x) = \frac{f_1(x)\pi}{f_1(x)\pi + f(x)(1-\pi)},
\end{equation*}
where $\pi = \frac{m}{m+n}$ is the proportion of labeled sample within the entire data. We treat $\pi$ as a known constant.

%Unfortunately, $C(x)$ depends on the knowledge of $f$ and $f_1$, which are difficult to estimate when dimension $p$ is large. In the next section, we show that $C(x)$ can be approximated by some empirical classifier $C_n(x)$.

Denote the distribution of probabilistic predictions for $P$ and $U$, respectively, as
\begin{align*}
G_L(t) &= P(C(X) \leq t | X \sim F_1),\\
G(t) &= P(C(X) \leq t | X \sim F).
\end{align*}
One can consider the two-component mixture model
\begin{align} \label{eq:model2}
G = \alpha^G G_s + (1-\alpha^G)G_L,
\end{align}
for $\alpha^G$ and $G_s$, which are again potentially non-identifiable. Define
\begin{align} \label{def:alpha0_G}
\alpha_0^G := \inf\left\lbrace \gamma \in (0,1]: \frac{G-(1-\gamma)G_L}{\gamma} \text{ is a c.d.f.}\right\rbrace.
\end{align}

\begin{Thm} \label{thm:change_problems}
$\alpha_0^G = \alpha_0$.
\end{Thm}
See Section \ref{prf:change_problems} for a proof. Theorem \ref{thm:change_problems} shows one can solve the p--dimensional MPE problem \eqref{def:alpha0} by solving the 1--dimensional MPE problem \eqref{def:alpha0_G}. In what follows we use $\alpha_0$ instead of $\alpha_0^G$ to simplify notation. 

%The difficulty of problem \eqref{def:alpha0} is due to the curse of dimensionality of non-parametric estimation. Usually, directly estimating $F, F_1, f$ or $f_1$ is almost impossible when $\dim(X) \geq 5$. However, in problem \eqref{def:alpha0_G}, life becomes easier since $G, G_L, g$ or $g_L$ is one-dimensional. The bottleneck of the problem \eqref{def:alpha0_G} would be to estimate $C(X)$. There are many well-calibrated classification algorithms like Random Forest \citep{Breiman2001} can be used to estimate $C(X)$. 

In practice, the classifier $C(X)$ is approximated by a trained model $C_n(X)$ on a given sample. For convenience, we assume the classifier $C_n(X)$ is trained using another independent sample $\mathcal{D}_n^\prime$. The $\mathcal{D}_n^\prime$ is omitted in the following to lighten notation. We require the approximated classifier to be a consistent estimator of the true classifier.

\begin{Assume} \label{ass:classifier_consistent}
We assume
\begin{align}
\E |C_n(X)-C(X)| = O\left( n^{-\tau} \right),
\end{align}
for some $\tau > 0$.
\end{Assume}

Such convergence results have been proven for a variety of probabilistic classifiers, including variants of Random Forest \citep{Biau2012}. Define

%For example, Random Forest model by Breiman is shown to be consistent \citep{Biau2012} in $L^2$ distance, under assumptions on the form of $C(x)$ and density on $X$. See Theorem \ref{thm:Biau2012} for detailed conditions and description. If we have $L^2$-consistency, then immediately Assumption \ref{ass:classifier_consistent} holds. In practice, we recommend to use Random Forest, which is also numerically argued to be well-calibrated \citep{Dankowski2016, Bostrom2008}.

  \begin{align*}
    G_{L,n}(t) &:= \frac{1}{m}\sum_{i=1}^m \ind{C_n(X_{L,i}) \leq t}, \\
    G_n(t) &:= \frac{1}{n}\sum_{i=1}^n \ind{C_n(X_{i}) \leq t}.
  \end{align*}
Intuitively, $G_{L,n}$ and $G_n$ are approximate empirical distribution functions of $G_L$ and $G$ respectively. The approximation is due to the fact that $C$ is estimated with $C_n$. Thus we would expect Glivenko-Cantelli and Donsker properties for $G_{n}(t)$ and $G_{L,n}(t)$. However problems can arise when $C(X)$ is not continuous. Essentially convergence in probability for $C(X)$, implied by Assumptions \ref{ass:classifier_consistent}, only implies convergence of distribution functions at points of continuity. By assuming $G_L$ and $G$ possess densities, we can obtain uniform convergence of distribution functions.

\begin{Assume} \label{ass:classifier_continuous}
We assume that $G$ and $G_L$ are absolutely continuous and have bounded density functions $g$ and $g_L$. 
\end{Assume}

\begin{Thm} \label{thm:CD_donsker} 
Under Assumption \ref{ass:classifier_consistent} and \ref{ass:classifier_continuous}, for $\beta = min(\tau/3,1/2)$
    \begin{align*}
      &n^{\beta}(G_{L,n}(t) - G_{L}(t)) \text{ is } O_P(1),\\
      &n^{\beta}(G_{n}(t) - G(t)) \text{ is } O_P(1),
    \end{align*}
where both $O_P(1)$ are uniform in $t$.
\end{Thm}

See Section \ref{prf:CD_donsker} for a proof. The result from Theorem \ref{thm:CD_donsker} is the key step in showing consistency of our $\alpha_0$ estimators in the following sections.

\section{Estimation of $\alpha_0$}
\label{sec:estimation_alpha0}

We generalize a one--dimensional MPE method of Patra and Sen \cite{patra2015estimation} to the PU learning problem. We term the method C-PS to emphasize the fact that the method developed by Patra and Sen is applied to the output of a classifier. Then we generalize a one--dimensional method of Storey \cite{Storey2002} to the PU learning problem. We term the method C-ROC because the ROC method developed in \cite{Blanchard2010} and \cite{scott2015rate} can be viewed as a variant of the Story's \cite{Storey2002} original idea.

\subsection{C-PS}
\label{subsec:patra/sen}

\cite{patra2015estimation} remove as much of the $G_{L,n}$ distribution from $G_n$ as possible, while ensuring that the difference is close to a valid cumulative distribution function. We briefly review the idea and provide theoretical results to support use of this procedure in the PU learning problem. See \cite{patra2015estimation} for a fuller description of the method in the one--dimensional case.

For any $\gamma \in (0,1]$ define
  \begin{equation*}
    \widehat{G}_{s,n}^\gamma = \frac{G_n - (1-\gamma)G_{L,n}}{\gamma}.
  \end{equation*}
  If $\gamma \geq \alpha_0$, $\widehat{G}_{s,n}^\gamma$ will be a valid c.d.f. (up to sampling uncertainty) while the converse is true if $\gamma < \alpha_0$. Find the closest valid c.d.f. to $\widehat{G}_{s,n}^\gamma$ defined as 
\begin{equation}
\label{eq:iso_reg}
\check{G}_{s,n}^\gamma = \argmin{\text{all c.d.f. } W(t)} \int \left( \widehat{G}_{s,n}^\gamma(t) - W(t) \right)^2 d G_n(t).
\end{equation}
Isotonic regression is used to solve Equation \ref{eq:iso_reg}. Measure the distance between two c.d.f $W_1$ and $W_2$ as
\begin{equation*}
  d_n(W_1,W_2)  = \sqrt{\int \left(W_1(t)-W_2(t) \right)^2d G_n(t)}. 
\end{equation*}
If $d_n(\widehat{G}_{s,n}^\gamma,\check{G}_{s,n}^\gamma) \approx 0$, then $\alpha_0 \leq \gamma$ where the level of approximation is a function of the estimation uncertainty and thus the sample size. Given a sequence $c_n$ define
\begin{equation*}
  \widehat{\alpha}_0^{c_n} = \inf \left\lbrace \gamma \in (0,1]: \gamma d_n(\widehat{G}_{s,n}^\gamma,\check{G}_{s,n}^\gamma) \leq \frac{c_n}{n^{\beta-\eta}}\right\rbrace
\end{equation*}
where $\eta \in (0, \beta)$ is a constant and the rate $\beta$ is from Theorem \ref{thm:CD_donsker}.

%Once we have estimators $G_n$ and $G_{L,n}$, in order to make $\widehat{G}_{s,n}^{\gamma}$ (Equation \eqref{eq:Gs_gamma}) a valid c.d.f. we have to choose $\gamma$ carefully. With small $\gamma$ too much of $G_{L,n}$ is removed so that $\widehat{G}_{s,n}^{\gamma}$ might not be non-decreasing; while with large $\gamma$, we will overestimate $\alpha_0$. Isotonic regression, which is utilized by \cite{patra2015estimation}, can be used to score the decreasing part of $\widehat{G}_{s,n}^{\gamma}$ with selected $\gamma$. Mathematically, isotonic regression on $\widehat{G}_{s,n}^{\gamma}$ is formulated in equation \eqref{eq:iso_reg}, which finds a non-decreasing function $W$ that minimizes the $L^2$-distance with $\widehat{G}_{s,n}^{\gamma}$. Numerically, equation \eqref{eq:iso_reg} can be efficiently optimized with complexity of $O(n)$ by pool adjacent violators algorithm (see \cite{Grotzinger1984}). 

\begin{Thm} \label{thm:PS_consistency} 
Under Assumptions \ref{ass:classifier_consistent} and \ref{ass:classifier_continuous},
if $c_n = o(n^{\beta-\eta})$ and $c_n \rightarrow \infty$, then $\widehat{\alpha}_0^{c_n} \xrightarrow{p} \alpha_0$.
\end{Thm}

The proof, contained in Section \ref{prf:PS_consistency}, is a generalization of results in \cite{patra2015estimation} which accounts for the fact that both $G_n$ and $G_{L,n}$ are estimators. While Theorem \ref{thm:PS_consistency} provides consistency, there are a wide range of choices of $c_n$. \cite{patra2015estimation} showed that $\gamma d_n(\widehat{G}_{s,n}^\gamma,\check{G}_{s,n}^\gamma)$ is convex, non-increasing and proposed letting $\widehat{\alpha}_0$ be the $\gamma$ that maximizes the second derivative of $\gamma d_n(\widehat{G}_{s,n}^\gamma,\check{G}_{s,n}^\gamma)$. We use this implementation in our numerical work in Section \ref{sec:experiments}.

\subsection{C-ROC} 
\label{subsec:storey}

Recalling the definitions of $G$, $G_s$, and $G_L$ from Section \ref{sec:identify}, note
\begin{equation*}
  G(t) = \alpha G_s(t) + (1-\alpha)G_L(t) \leq \alpha + (1-\alpha)G_L(t)
\end{equation*}
for all $t$. Thus for any $t$ such that $G_L(t)\neq 1$ we have
\begin{equation*}
  k(t) \equiv \frac{G(t) - G_L(t)}{1-G_L(t)} \leq \alpha.
\end{equation*}
In the FDR literature, $G_L$ is the distribution of the test statistic or p--value under the null hypothesis and is generally assumed known. Thus only $G$ must be estimated, usually with the empirical cumulative distribution function.  \cite{Storey2002} proposed an estimator for $k(t)$ at fixed $t$ (Equation 6) and determined a bootstrap method to find the $t$ which produces the best estimates of the FDR.

The PU problem is more complicated in that one must estimate $G$ and $G_L$. However the structure of $G$ and $G_L$ enables one to estimate the identifiable parameter $\alpha_0$. Specifically with $t^* = \inf \{ t : G_L(t) \geq 1\}$ we have
\begin{equation}
  \label{eq:sugg}
  \lim_{t \uparrow t^*} k(t) = \alpha_0.
\end{equation}
See Lemma \ref{lem:ktlim} for a proof. This result suggests estimating $\alpha_0$ by substituting the empirical estimators of $G_n$ and $G_{L,n}$ into Equation \ref{eq:sugg} along with a sequence $\widehat{t}$ which is converging to the (unknown) $t^*$. Such a sequence $\widehat{t}$ must be chosen so that the estimated denominator $1-\widehat{G}_{L,n}(\widehat{t})$ is not converging to $0$ too fast (and hence too variable). For $\widehat{t}$ we use a quantile of the empirical c.d.f. which is converging to $1$, but at a rate slower than the convergence of the empirical c.d.f.. For some $q \in (0,\beta)$, define

\begin{equation*}
  \widehat{t} = \inf \{ t : G_{L,n}(t) \geq 1-n^{-q}\} - n^{-1}.
\end{equation*}
The $n^{-1}$ term in $\widehat{t}$ avoids technical complications.
\begin{Thm}
  \label{thm:storey}
  Under Assumptions \ref{ass:classifier_consistent} and \ref{ass:classifier_continuous}
  \begin{equation*}
k_n(\widehat{t}) \equiv \frac{G_n(\widehat{t}) - G_{L,n}(\widehat{t})}{1-G_{L,n}(\widehat{t})} \xrightarrow{P} \alpha_0.
  \end{equation*}
\end{Thm}
See Section \ref{prf:storey} for a proof.

\begin{figure}[t]
\centering
\includegraphics[width=1\textwidth,height=6.5cm]{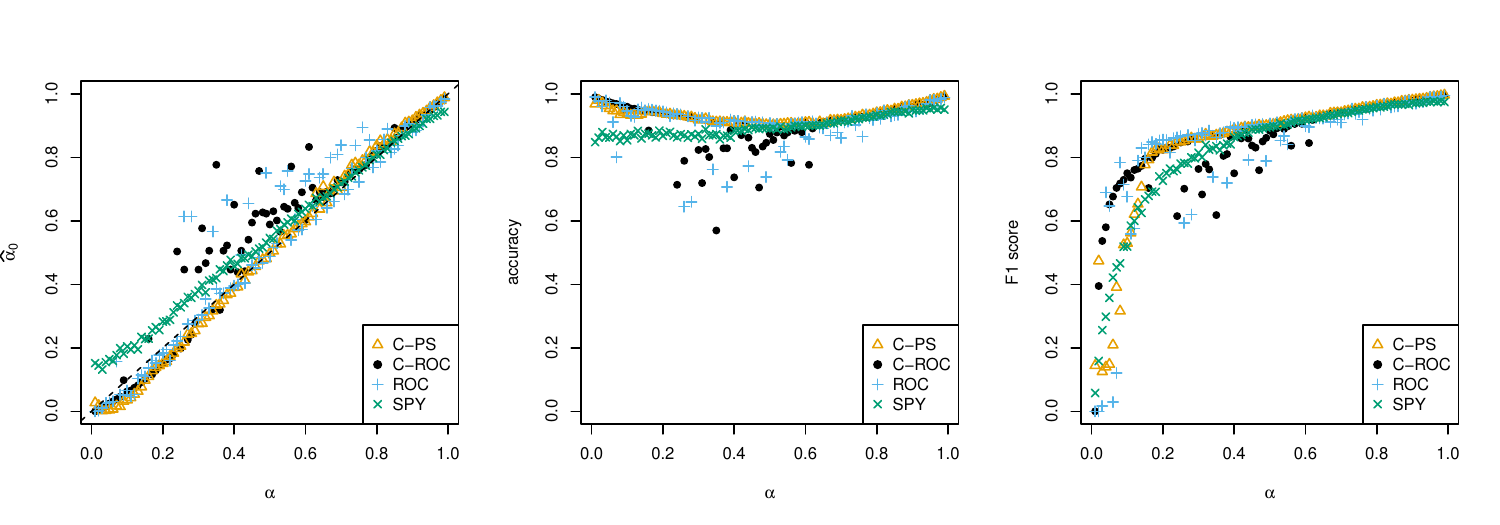}
{\caption{Comparison of methods with different $\alpha$ values. On the x-axis, $\alpha$ varies from 0.01 to 0.99 by step size 0.01. The left plot displays the estimates of the lower bound $\alpha_0$. The middle plot displays the accuracy of classifying observations in $U$. The right plot displays the F1 score of the classifications.}\label{fig:figA}}
\end{figure}

\subsubsection{Connection with ROC Method}
\label{sec:roc_comp}

The ROC method of \cite{scott2015rate} solves a generalization of the PU learning problem in which the positive set contains mislabeled data. When this method is specialized to the case of no misclassification in the labeled data i.e. the PU learning problem, it becomes a variant of the \cite{Storey2002} method with a particular cutoff value $t$. Specifically, define the true ROC curve by the parametric equation
\begin{equation*}
%  R: [0,1] \rightarrow (G_L(t),G(t)).
  \{(G_L(t),G(t)) : t \in [0,1]\}.
\end{equation*}
\cite{scott2015rate} (Proposition 2) showed that $\alpha_0$ is the supremum of one minus the slope between (1,1) and any point on the ROC curve.\footnote{\cite{scott2015rate} estimated $\kappa = 1 - \alpha$. We have modified the ROC method notation to reflect the $\alpha$ notation used in this work.} This is equivalent to the Storey method of \cite{Storey2002} because
\begin{align*}
  \alpha_0 &= \sup_t 1 - \frac{1- G(t)}{1-G_L(t)}\\
  &= \sup_t \frac{G(t) - G_L(t)}{1-G_L(t)}\\
  &= \sup_t k(t).
\end{align*}
The true ROC curve is not known, so $\alpha_0$ cannot be computed directly from this expression. \cite{Blanchard2010} found a consistent estimator and \cite{scott2015rate} determined rates of convergence using VC theory. For application to data, \cite{scott2015rate} splits the labeled and unlabeled data sets in half, constructs a kernel logistic regression classifier on half the data, and estimates the slope between (1,1) and a discrete set of points on the ROC curve. The $\alpha_0$ estimate is the supremum of 1 minus each of these slopes. Thus we see that the ROC method and earlier methods developed in the FDR literature are in the same family of $\alpha$ estimation strategies. Choosing a $t$ in the Storey approach is equivalent to choosing a point on the ROC curve.

\subsubsection{Practical Implementation}

We consider two implementations of these ideas. The method of \cite{scott2015rate}, using a kernel logistic regression classifier and a PU training--test set split to estimate tuning parameters, is referred to as ``ROC.'' To facilitate comparison with C-PS, we consider another version with a Random Forest classifier using out--of--bag probabilities to construct the ROC curve. We call this method C-ROC.

\section{Numerical Experiments}
\label{sec:experiments}

To illustrate the proposed methods we carry out numerical experiments on simulated \textit{waveform} data and a real protein signaling data set \textit{TCDB-SwissProt}. We compare the performance of the three methods (C-PS, C-ROC and ROC) discussed in Section \ref{sec:estimation_alpha0} and the SPY method. With the SPY method, once the classifications (``positive'' or ``negative'') in set $U$ are made, we use the proportion of ``negative'' cases as an approximation of $\alpha_0$. For the C-ROC and C-PS methods \citep{Breiman2001}, we use Random Forest to construct $C_n(\cdot)$.

\subsection{Waveform Data}

\begin{figure}[t]
  \begin{center}
    \begin{includegraphics}[height=12cm,width=15cm]{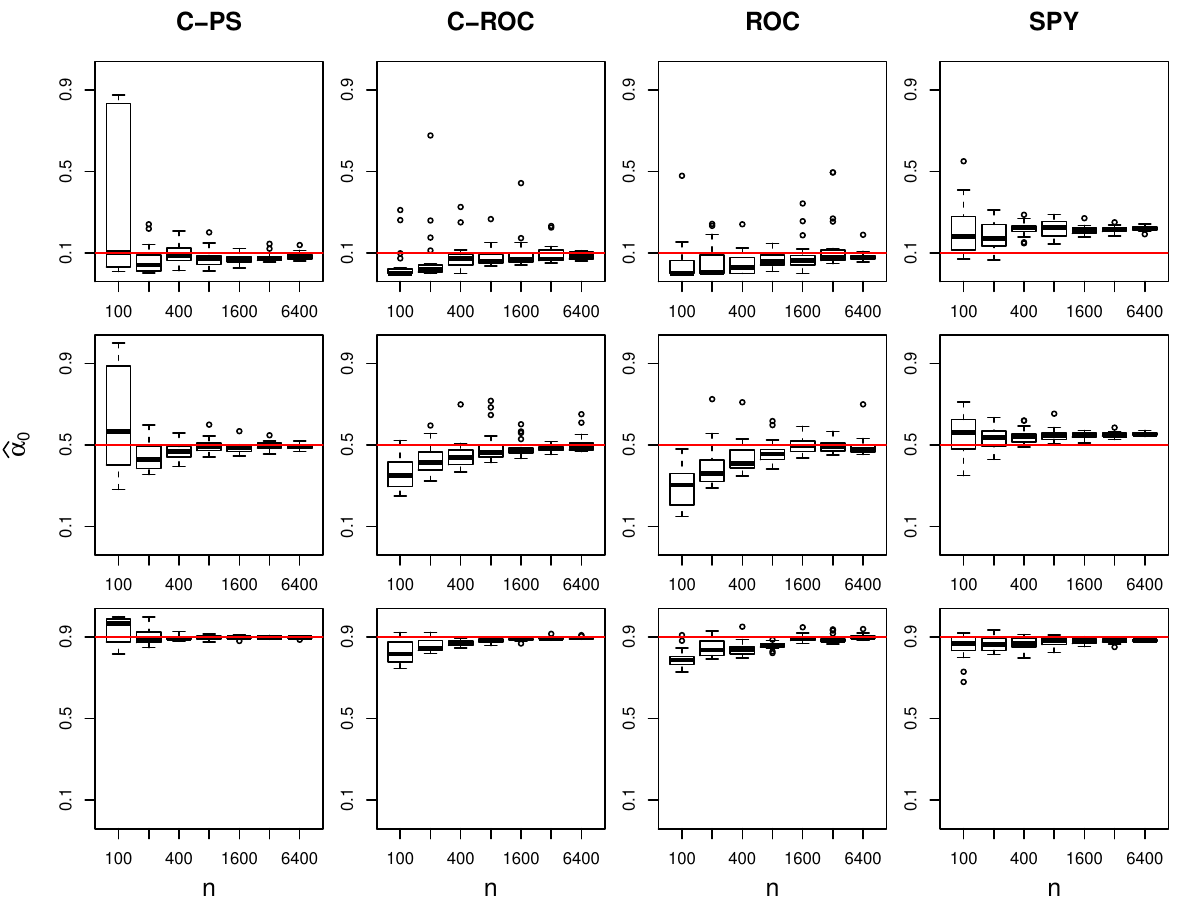}
      \caption{Comparison of methods with different sample sizes. The red solid horizontal lines represent the true $\alpha$ (0.1,0.5,0.9). The range for all y-axes is [0, 1] from bottom to top. The unlabeled sample size $n$ varies with $100\times 2^j ( j=0,\ldots, 6)$. Each boxplot summarizes 20 repeated estimates $\widehat{\alpha}_0$ for each $(n,\alpha)$ pair.\label{fig:comparison_varying_samplesize}}
    \end{includegraphics}
  \end{center}
\end{figure}

We simulate observations from the \textit{waveform} data set using the R-package \textit{mlbench} \citep{mlbench}. The \textit{waveform} data is a binary classification problem with 21 features. We fix $\pi = 0.5$ for all simulations.

\subsubsection{Varying $\alpha$}

We vary $\alpha$ from 0.01 to 0.99 in Model \eqref{eq:model1} in increments of $0.01$. For each $\alpha$ the sample sizes are fixed at $m=n=3000$. At each $\alpha$ we run the methods described to estimate $\alpha$ and classify observations in $U$. Results are shown in Figure \ref{fig:figA}. SPY produces inflated $\alpha$ estimates at low $\alpha$ (left panel) and has the the worst overall classification performance (center panel). Both C-ROC and ROC have substantial variability for $\alpha$ near 0.5. Overall, C-PS appears to be the best method.

\subsubsection{Varying Sample Size}

We empirically examine consistency and convergence rates of the methods by estimating $\alpha$ at increasing sample sizes, keeping the number of labeled and unlabeled observations equal, i.e. $n=m$. In Figure \ref{fig:comparison_varying_samplesize}, every method is repeated 20 times for each $(n, \alpha)$ pair. The 20 $\alpha_0$ estimates are displayed as a boxplot, which show estimator bias and variance. We see that all methods, except SPY, appear consistent under different settings ($\alpha = 0.1, 0.5, 0.9$). The estimators may have substantial bias at small $n$. C-PS struggles at $n=100$, but has the best overall performance, followed by C-ROC and ROC.

\subsubsection{Single Feature $\alpha_0$ Estimation}
\label{sec:single_feat}

\begin{figure}[ht]
  \begin{center}
    \begin{includegraphics}[height=5cm,width=8.5cm]{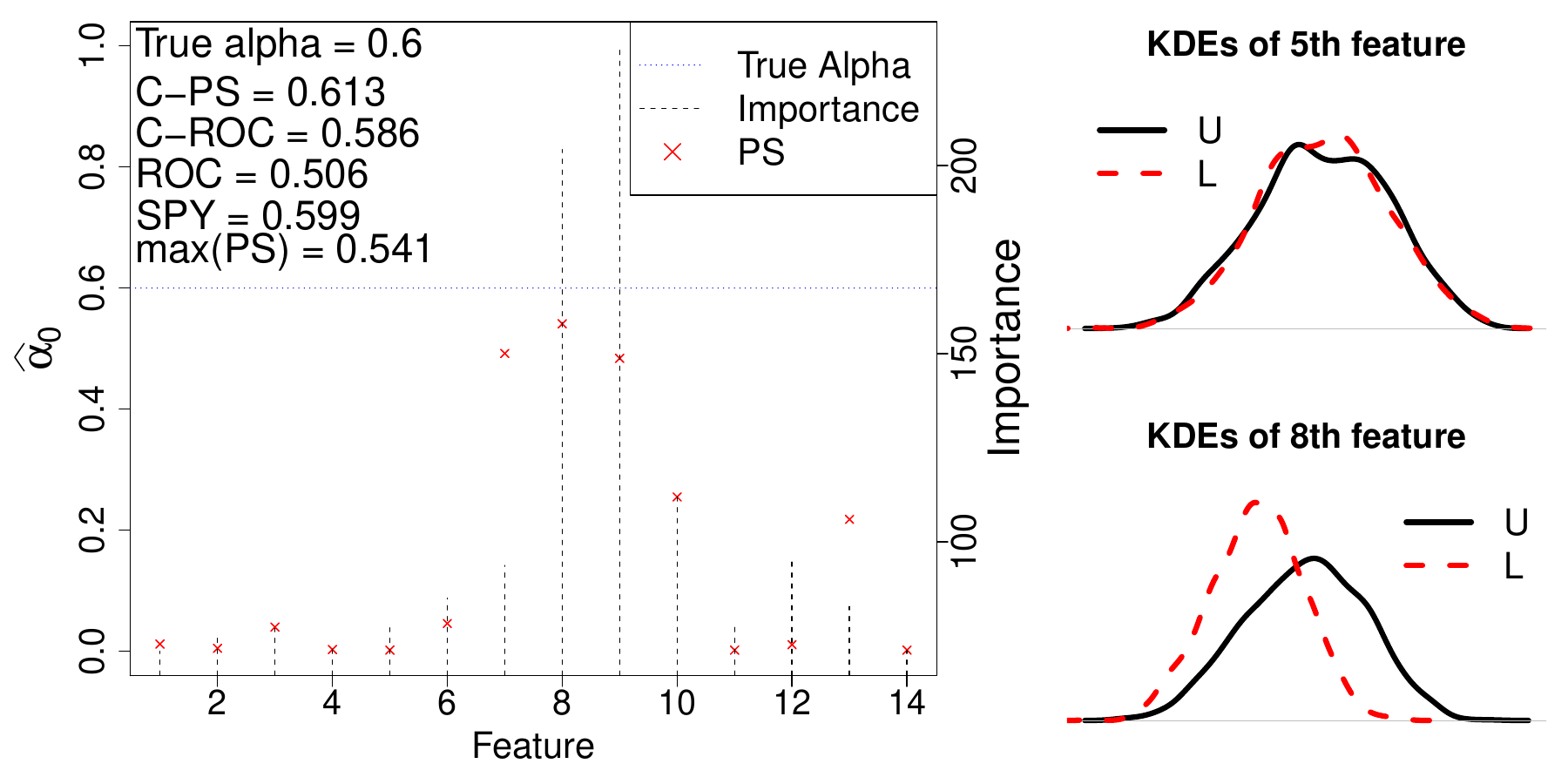}
      \caption{Estimation of $\alpha_0$ using individual features. In the left panel the horizontal blue dash line is the true $\alpha$ ($=0.6$), the vertical black dashed lines are the feature importances (right y--axis), and the red cross symbol are the $\alpha_0$ estimate using the Patra/Sen procedure on a single feature (left y--axis). The right panels are kernel density estimates of unlabeled (U) and labeled (L) data for features 5 and 8.
      \label{fig:comparison_varying_feature}}
    \end{includegraphics}
  \end{center}
\end{figure}

One approach to solving the multidimensional PU learning problem is to estimate $\alpha$ separately using each feature. If $X_i \in \mathbb{R}^p$, this results in $p$ estimates $\widehat{\alpha}_0^1, \ldots,\widehat{\alpha}_0^{p}$ of the parameter $\alpha$. Each of these is an estimated lower bound on $\alpha$. Thus a naive estimate of $\alpha_0$ is $\max(\widehat{\alpha}_0^1, \ldots,\widehat{\alpha}_0^{p})$. This approach ignores the correlation structure among features.

%using using one-dimensional method, for example patra/sen procedure, is working on individual feature separately and then combining the results together. In Figure \ref{fig:comparison_varying_feature} for example, each feature $\alpha$ estimate, $\widehat{\alpha}_0^1, \ldots,\widehat{\alpha}_0^{14}$, is an estimated lower bound on $\alpha$. So we can take the maximum of these values, $\widehat{\alpha}_0 = \max\{\widehat{\alpha}_0^1, \ldots,\widehat{\alpha}_0^{14}\}$. 

Using the \textit{waveform} data, we compare this strategy to the multi--dimensional classifier approach. To make the problem challenging we select the 14 weakest features, defined as having the lowest Random Forest importance scores. We apply the Patra--Sen one--dimensional method to obtain individual feature $\alpha_0$ estimates. The results are summarized in Figure \ref{fig:comparison_varying_feature}. Feature importance matches well with the performance of the $\alpha$ estimates. On the right panels of Figure \ref{fig:comparison_varying_feature}, we see that feature 5 is not useful because there is little difference between the unlabeled and labeled samples, leading to a feature based $\alpha$ estimate of approximately $0.012$. In contrast, feature 8 is better in that it gives an alpha estimate of approximately $0.542$. The SPY, C-ROC, and C-PS methods all perform better than the individual feature estimates (upper left of Figure \ref{fig:comparison_varying_feature}).

\subsection{Protein Signaling}

The transporter classification database (TCDB) \citep{tcdb2006}, here the $P$ set, consists of 2453 proteins involved in signaling across cellular membranes. It is desirable to add proteins to this database from unlabeled databases which contain a mixture of membrane transport and non--transport proteins. \cite{elkan2008learning} and \cite{das2007} manually identified 348 of the 4906 proteins as being related to transport in the SwissProt \citep{swissprot2003} database. We treat the SwissProt data as the unlabeled set $U$ for which we have ground truth $\alpha = (4906 - 348) / 4906 \approx 0.929$. Information from protein description documents are used as features including function, subcellular location, alternative products, and disease. In total there are $p=741$ features. We fit models with both the original feature set and with $2p$ and $10p$ features where all additional features are simulated by randomly selecting one of the original $p$ features and permuting its values among the observations. So for $10p$, $p$ features are original (and potentially useful for classifying observations) while $9p$ of the features are simulated noise. Since $10p > n+m$ (total training set size), this represents a high dimensional setting for estimating $\alpha$.

We compare C-PS with single feature PS for $p$, $2p$, and $10p$ features. A common strategy in high dimensional classification problems is to perform feature screening prior to classifier construction. Since the features are all binary, we screen features based on p--values from univariate chi-squared tests (Fisher exact when any $2 \times 2$ table cell counts are less than 10). We test the methods after screening for the top $k=500,200,100,50,10,1$ features with smallest p--values. The C-PS method is applied directly as described earlier on the $k$ best features. For the single feature PS method, after screening for $k$ features, the PS method is applied to all $k$ features individually and the largest estimate is taken as an estimate of $\alpha$. As explained in Section \ref{sec:single_feat}, for each feature the PS method is an estimate of a lower bound on $\alpha$, thus taking the maximum of these estimated lower bounds is sensible. 

Table \ref{tab:tcdb_swissprot} shows $\alpha$ estimates for each number of features and each $k$. First consider the two $p$ feature columns representing the C-PS and single feature PS methods. C-PS with $p$ features produces estimates which are high for $k=50$ to $500$ features, nearly correct for $k=10$ features, and biased quite low for $k=1$ feature. There appears to be some overfitting with large $k$, but extreme screening to $k=1$ results in a loss of information and a poor lower bound on $\alpha$. Single feature PS produces estimates that are too high for $k=500$ through $10$ and too low for $p=1$. It is either worse or no better than C-PS at each $k$. Single feature PS with $k=1$ represents choosing the best feature (based on p-values) and then applying the PS method. There is not sufficient information in this single best feature to obtain a good lower bound on $\alpha$. The behavior of PS overestimating $\alpha$ at $k=10$ through $500$ features is due to the sensitivity of taking the maximum of single feature PS estimates. Even a single large estimate on one feature results in an overall estimate which is too high. The natural way to correct this is to choose a smaller set of $k$ features, but with $k=1$ there is not sufficient information in this single feature to obtain a good lower bound. In contrast, C-PS effectively pools information across multiple features to produce improved estimates.

The results are quite consistent across $p$, $2p$ and $10p$. This is due to the fact that feature screening retains a very similar set of features regardless of the number of noisy features added to the data set. For example, with $k=500$, the $p$ and $10p$ feature models retain 316/500 of the same features while for $k=10$ they retain 10/10 of the same features. Thus the estimation methods (C-PS and PS) produce similar $\alpha$ estimates with $p$ and $10p$ features. This suggests that feature prescreening combined with C-PS can be an effective $\alpha$ estimation strategy in high--dimensional settings with many pure noise features.

% latex table generated in R 3.5.1 by xtable 1.8-3 package
% Fri Nov 15 11:04:49 2019
\begin{table}[ht]
\centering
\begin{tabular}{c|cc|cc|cc}
  & \multicolumn{2}{c}{p} & \multicolumn{2}{c}{2p} & \multicolumn{2}{c}{10p} \\
 \hline
k & C-PS & PS & C-PS & PS & C-PS & PS \\ 
  \hline
500 & 0.97 & 1.00 & 0.96 & 1.00 & 0.97 & 1.00 \\ 
  200 & 0.97 & 1.00 & 0.97 & 1.00 & 0.97 & 1.00 \\ 
  100 & 0.96 & 0.99 & 0.96 & 0.99 & 0.96 & 0.99 \\ 
  50 & 0.95 & 0.99 & 0.95 & 0.99 & 0.95 & 0.99 \\ 
  10 & 0.93 & 0.98 & 0.93 & 0.98 & 0.93 & 0.98 \\ 
  1 & 0.77 & 0.77 & 0.77 & 0.77 & 0.77 & 0.77 \\ 
   \hline
\end{tabular}
\caption{$\alpha$ estimates from C-PS and single feature PS with p, 2p, and 10p features when prescreening the top k features. True $\alpha \approx 0.93$.} 
\label{tab:tcdb_swissprot}
\end{table}

\section{Conclusion} \label{sec:conclusion}

In this paper we proposed a framework for estimating the mixture proportion and classifier in the PU learning problem. We implemented this framework using two estimators from the FDR literature, C-PS and C-ROC. The framework has the power to incorporate other one-dimensional MPE procedures, such as \cite{meinshausen2006estimating}, \cite{genovese2004stochastic}, \cite{langaas2005}, \cite{efron2007}, \cite{jin2008}, \cite{cai2010} or \cite{nguyen2014}. More generally we have strengthened connections between the classification--machine learning literature and the multiple testing literature by constructing estimators using ideas from both communities. Potential directions for future research include generalizing results to the case where the labeled data contains some mislabeling of observations, relaxing Assumption A, and developing methods to handle cases where labeled and unlabeled data sets sizes are substantially different (class imbalance).

\section*{Supplementary Materials}
\label{sec:supp}

%\begin{description}
%\item[\texttt{R} code for reproducing results:]
  \texttt{R}--code and data needed for reproducing results in this work are available online at \url{github.com/zflin/PU_learning}. % The \texttt{R} packages used in this project are \textit{mlbench} \citep{mlbench}, \textit{devtools} \citep{devtools}, \textit{Iso} \citep{iso}, \textit{randomForest} \citep{randomforest}, \textit{MASS} \citep{mass}, \textit{pdist} \citep{pdist}, \textit{LiblineaR} \citep{liblinear}, \textit{foreach} \citep{foreach}, \textit{doParallel} \citep{doparallel}, and \textit{xtable} \citep{xtable}.
%\end{description}

\baselineskip=14pt
\section*{Acknowledgments}

Part of this work was completed while the authors were research fellows in the Astrostatistics program at the Statistics and Applied Mathematical Sciences Institute (SAMSI) in Fall 2016. The authors gratefully acknowledge SAMSI's support.

\subsection*{Conflict of interest}

The authors declare no potential conflict of interests.

\section*{Supporting information}

The following supporting information is available as part of the online article:

\noindent
\textbf{Technical Notes}
{Proofs of all Theorems in this work.}

%\nocite{*}% Show all bib entries - both cited and uncited; comment this line to view only cited bib entries;
%\bibliography{wileyNJD-AMS}%

\bibliographystyle{abbrvnat}
\bibliography{refs}

\clearpage\pagebreak\newpage
%\pagestyle{fancy}
%\fancyhf{}
%\rhead{\bfseries\thepage}
%\lhead{\bfseries Technical Notes}
\begin{center}
{\LARGE{\bf Technical Notes for\\ {\it A Flexible Procedure for Mixture Proportion Estimation in Positive--Unlabeled Learning}}}
\end{center}

\setcounter{equation}{0}
\setcounter{page}{1}
\setcounter{table}{1}
\setcounter{section}{0}
\renewcommand{\theequation}{A.\arabic{equation}}
\renewcommand{\thesection}{A.\arabic{section}}
\renewcommand{\thesubsection}{A.\arabic{section}.\arabic{subsection}}
\renewcommand{\thepage}{A.\arabic{page}}
\renewcommand{\thetable}{A.\arabic{table}}
\baselineskip=17pt

\section{Proof of Theorems}
\label{sec:technical_notes}

\subsection{Proof of Theorem \ref{thm:change_problems}}
\label{prf:change_problems}

Equivalently, we are trying to prove
\begin{align}
\frac{G-(1-\gamma)G_L}{\gamma} \text{ is a c.d.f.} \Leftrightarrow \frac{F-(1-\gamma)F_1}{\gamma} \text{ is a c.d.f.}
\end{align}
Sufficient to show
\begin{align}
G-(1-\gamma)G_L \text{ non-decreasing} \Leftrightarrow& f - (1-\gamma)f_1 \geq 0\\
&\text{ with probability 1} \nonumber.
\end{align}
First we show $\Leftarrow$. Consider any $t_2 > t_1$. Then
\begin{align*}
  &\left(G(t_2) - (1-\gamma)G_L(t_2)\right) - \left(G(t_1) - (1-\gamma)G_L(t_1)\right) \\
  &= \int_{\{x: C(x) \in (t_1,t_2]\}} \underbrace{f(x) - (1-\gamma)f_1(x)}_{\geq 0 \text{ by assumption}} d\mu(x)\\
    &\geq 0.
\end{align*}
Now we show $\Rightarrow$ by proving the contrapositive. By assumption there exists
\begin{equation*}
A = \{x: f(x) - (1-\gamma)f_1(x) < 0\}
\end{equation*}
such that $P(A) > 0$. Further we have
\begin{align*}
  A &= \left\{x : (1-\gamma)\frac{(1-\pi)}{\pi} > \frac{f(x)}{f_1(x)}\frac{(1-\pi)}{\pi}\right\}\\
  &= \left\{x : \underbrace{\frac{1}{1 + (1-\gamma)\frac{(1-\pi)}{\pi}}}_{\equiv t^*} < C(x)\right\}.
\end{align*}
So
\begin{align*}
&\left(G(1) - (1-\gamma)G_L(1)\right) - \left(G(t^*) - (1-\gamma)G_L(t^*)\right) \\
&=  \int_{A=\{x: C(x) > t^*\}} f(x) - (1-\gamma)f_1(x) d\mu(x)\\
    &< 0.
\end{align*}

\subsection{Proof of Theorem \ref{thm:CD_donsker}}
\label{prf:CD_donsker}
\begin{align*}
  n^{\beta}(G_{n}(t) - G(t)) &= \frac{n^\beta}{n^{1/2}} \underbrace{n^{1/2}\left(G_{n}(t) - \E[\ind{C_n(X) \leq t} \vert C_n]\right)}_{\equiv \mathbb{R}_n(t)}\\
  &+ n^\beta\underbrace{\left( \E[\ind{C_n(X) \leq t} \vert C_n] - G(t)\right)}_{\equiv \mathbb{Q}_n(t)}
\end{align*}
We now show that $\mathbb{R}_n(t)$ and $\mathbb{Q}_n(t)$ are $O_P(1)$ uniformly in $t$. Together these facts show the expression is $O_P(1)$ uniformly in $t$.

\noindent
\underline{\textbf{$\mathbb{R}_n(t)$:}} Note
\begin{equation*}
  \mathbb{R}_n(t) = \sqrt{n}\left(\frac{1}{n} \sum_{i=1}^n \ind{C_n(X_i) \leq t} - \E[\ind{C_n(X) \leq t} \vert C_n]\right).
\end{equation*}
By the DKW inequality
\begin{equation*}
  P(||\mathbb{R}_n||_{\infty} > x \big\vert C_n) \leq 2e^{-2x^2}.
\end{equation*}
Thus $||\mathbb{R}_n||_{\infty}$ is $O_P(1)$.%%  Note $\mathbb{R}_n$ is a convolution (in $t$) of $\mathbb{Q}_n$ and a normal density, specifically
%% \begin{equation*}
%%   \mathbb{R}_n(t) = \int \mathbb{Q}_n(t - s)\left(\frac{1}{\sigma_n}\phi\left(\frac{s}{\sigma_n}\right)\right) ds
%% \end{equation*}
%% where $\phi$ is the standard normal density. Thus $||\mathbb{R}_n||_{\infty} \leq ||\mathbb{Q}_n||_{\infty}$ and hence $||\mathbb{R}_n||_{\infty}$ is $O_P(1)$.

\noindent
\underline{\textbf{$\mathbb{Q}_n(t)$:}} We have
\begin{align*}
  \mathbb{Q}_n(t) &= \E[\overbrace{(\ind{C_n(X) \leq t} - \ind{C(X) \leq t})}^{\equiv T_n}\big\vert C_n]\\
  &\leq \underbrace{|E[T_n \ind{|C(X) - t| \leq \epsilon_n} \big\vert C_n ]|  }_{B_1}\\
  &+ \underbrace{|E[T_n \ind{|C(X) - t| > \epsilon_n}\ind{|C(X)-C_n(X)|<\epsilon_n} \big\vert C_n ]|}_{B_2} \\
  & + \underbrace{|E[T_n \ind{|C(X) - t| > \epsilon_n}\ind{|C(X)-C_n(X)|>\epsilon_n} \big\vert C_n ]|}_{B_3}
\end{align*}
Noting that $|T_n| \leq 1$ and $C_n$ is independent of $C(X)$, we have
\begin{equation*}
  B_1 \leq P(|C(X) - t| \leq \epsilon_n) \leq 2\epsilon_n \sup_{t} g(t)
\end{equation*}
where $g$ is the density of $C(X)$, which exists and is bounded by Assumptions \ref{ass:classifier_continuous}. $B_2$ is $0$ because $T_n=0$ whenever the indicator functions in $B_2$ are both $1$. Finally noting $B_3 \leq \E[\ind{|C(X) - C_n(X)| > \epsilon_n}|C_n]$ and using Markov's inequality twice, we have
\begin{align*}
  P(B_3 > r_n) &\leq P(\E[\ind{|C(X) - C_n(X)| > \epsilon_n}|C_n] > r_n)\\
  &\leq \frac{P(|C(X) - C_n(X)| > \epsilon_n)}{r_n}\\
  &\leq \frac{E[|C_n(X) - C(X)|]}{\epsilon_nr_n}.
\end{align*}
Setting $\epsilon_n = n^{-\tau/3}$, $r_n = n^{-\tau/3}$, and $\beta = \tau/3$ achieves the desired result. Identical arguments hold for showing $n^{\beta}(G_{L,n}(t) - G_L(t))$ is $O_P(1)$ uniform in $t$.

\subsection{Proof of Theorem \ref{thm:storey}}
\label{prf:storey}

Since $  \widehat{t} = \inf \{ t : G_{L,n}(t) \geq 1-n^{-q}\} - n^{-1}$ and $0 < q < \beta$, we have
\begin{equation*}
  (n^\beta(1-G_{L,n}(\widehat{t})))^{-1} = \frac{n^q}{n^\beta} = o(1).
\end{equation*}
Recall  by Theorem \ref{thm:CD_donsker} we have
\begin{align*}
&n^{\beta}(G_{L,n}(t) - G_{L}(t)) \equiv d_L(t) = O_P(1)\\
&n^{\beta}(G_{n}(t) - G(t)) \equiv d(t) =  O_P(1)
\end{align*}
where this and subsequent $O_P$ and $o_P$ are uniform in $t$. We have
\begin{align*}
\frac{G_n(\widehat{t}) - G_{L,n}(\widehat{t})}{1-G_{L,n}(\widehat{t})} &= 
\frac{G(\widehat{t}) - G_L(\widehat{t})}{1-G_{L,n}(\widehat{t})} + \frac{n^{-\beta}(d_L(\widehat{t}) - d(\widehat{t}))}{1-G_{L,n}(\widehat{t})}\\
&= \underbrace{\left(\frac{1-G_{L}(\widehat{t})}{1-G_{L,n}(\widehat{t})}\right)}_{\equiv A}  \underbrace{\left(\frac{G(\widehat{t}) - G_L(\widehat{t})}{1-G_L(\widehat{t})}\right)}_{\equiv k(\widehat{t})}\\
&+ \underbrace{\frac{d_L(\widehat{t}) - d(\widehat{t})}{n^\beta(1-G_{L,n}(\widehat{t}))}}_{o_P(1)}.
\end{align*}
Note that
\begin{equation*}
  A = 1 + \frac{d_L(\widehat{t})}{n^\beta(1 - G_{L,n}(\widehat{t}))} = 1 + o_P(1).
\end{equation*}
Thus it is sufficient to show that $k(\widehat{t}) \rightarrow \alpha_0$. By Lemma \ref{lem:ktlim}, $k(t) \uparrow \alpha_0$ as $t \uparrow t^*$. We show that for any $\epsilon > 0$
\begin{equation*}
  P(\widehat{t} \in (t^*-\epsilon,t^*))\rightarrow 1.
\end{equation*}
Thus by the continuous mapping theorem, the estimator is consistent.

    \noindent
\underline{\textbf{Part 1:}} We show $P(t^* - \widehat{t} > \epsilon) \rightarrow 0$. By the definition of $t^*$, there exists $\gamma > 0$ such that $G_{L}(t^* - \epsilon/2) = 1 - \gamma$. We have
    \begin{align*}
    &P(t^* - \widehat{t} > \epsilon)\\
    &= P(G_{L,n}(t^* - \epsilon + n^{-1}) > G_{L,n}(\widehat{t} + n^{-1}))\\
    &\leq P(G_{L,n}(t^* - \epsilon + n^{-1}) > 1-n^{-q})\\
    &\leq \underbrace{P(G_L(t^* - \epsilon + n^{-1}) > 1-n^{-q} - \gamma / 2)}_{\equiv A} \\
    &+ \underbrace{P(|G_{L,n}(t^* - \epsilon + n^{-1}) - G_L(t^* - \epsilon + n^{-1})| > \gamma / 2)}_{\rightarrow 0 \text{ by Theorem \ref{thm:CD_donsker}}}.
    \end{align*}
$A \rightarrow 0$ because for sufficiently large $n$, $G_L(t^* - \epsilon + n^{-1}) \leq G_L(t^* - \epsilon/2) = 1 - \gamma < 1 - n^{-q} - \gamma/2$.

    \noindent
    \underline{\textbf{Part 2:}} We show $P(\widehat{t} \geq t^*) \rightarrow 0$. We have
    \begin{align*}
      P(\widehat{t} \geq t^*) &= P(G_{n,L}(\widehat{t} + n^{-1}) \geq G_{n,L}(t^* + n^{-1}))\\
      &= P(1 - n^{-q} \geq G_{n,L}(t^* + n^{-1}))\\
      &= P(1 - G_{n,L}(t^* + n^{-1}) \geq n^{-q})\\
      &= P(\underbrace{n^\beta(G_{L}(t^* + n^{-1}) - G_{n,L}(t^* + n^{-1}))}_{O_P(1) \text{ by Theorem \ref{thm:CD_donsker}}} \geq n^{\beta-q}).
    \end{align*}
Since $\beta > q$ we have the result.

\subsection{Proof of Theorem \ref{thm:PS_consistency}}
\label{prf:PS_consistency}
\begin{proof}
$\forall \epsilon > 0$, we need to show $P(|\widehat{\alpha}_0^{c_n} - \alpha_0| > \epsilon) \rightarrow 0$. Note
\begin{equation*}
P(|\widehat{\alpha}_0^{c_n} - \alpha_0 | > \epsilon) =  P(\widehat{\alpha}_0^{c_n} < \alpha_0 - \epsilon) + P(\widehat{\alpha}_0^{c_n} > \alpha_0 + \epsilon).
\end{equation*}
First we show that $P(\widehat{\alpha}_0^{c_n} < \alpha_0 - \epsilon) \rightarrow 0$. If $\alpha_0 \leq \epsilon$, then
\begin{align*}
P(\widehat{\alpha}_0^{c_n} < \alpha_0 - \epsilon) \leq P(\widehat{\alpha}_0^{c_n} < 0) = 0.
\end{align*}
If $\alpha_0 > \epsilon$, suppose we have $\widehat{\alpha}_0^{c_n} < \alpha_0 - \epsilon$, then by Lemma \ref{lem:PS_convex},
\begin{align*}
d_n(\widehat{G}_{s,n}^{\alpha_0-\epsilon}, \check{G}_{s,n}^{\alpha_0-\epsilon}) \leq \frac{c_n}{n^{\beta-\eta} (\alpha_0-\epsilon)}.
\end{align*}
The LHS of above converges to positive constant by Lemma \ref{lem:PS_distance}, while the RHS converges to zero by the choice of $c_n$, hence $P(\widehat{\alpha}_0^{c_n} < \alpha_0 - \epsilon) \rightarrow 0$.

Now we show that $P(\widehat{\alpha}_0^{c_n} > \alpha_0 + \epsilon) \rightarrow 0$. Suppose we have $\widehat{\alpha}_0^{c_n} > \alpha_0 + \epsilon$, then by Lemma \ref{lem:PS_convex},
\begin{align*}
n^{\beta-\eta}d_n(\widehat{G}_{s,n}^{\alpha_0+\epsilon}, \check{G}_{s,n}^{\alpha_0+\epsilon}) > \frac{c_n}{(\alpha_0-\epsilon)}.
\end{align*}
The LHS of above converges to zero by Lemmas \ref{lem:PS_distance} and \ref{lem:CVM_rate}, while the RHS diverges to infinity by the choice of $c_n$, hence $P(\widehat{\alpha}_0^{c_n} > \alpha_0 + \epsilon) \rightarrow 0$.
\end{proof}

\section{Lemmas}
\label{sec:lemmas}

    \begin{Lem}
      \label{lem:ktlim}
  $\lim_{t\uparrow t^*} k(t) = \alpha_0$.
    \end{Lem}
    \begin{proof}
      Define $\alpha_0' = \lim_{t\uparrow t^*} k(t)$.

      \noindent
      \textbf{Show $\alpha_0' \leq \alpha_0$:} By the definition of $\alpha_0$ there exists c.d.f. $G_{\alpha_0}$ such that
  \begin{align*}
    G(t) &= \alpha_0G_{\alpha_0}(t) + (1-\alpha_0)G_L(t)\\
    &\leq \alpha_0 + (1-\alpha_0)G_L(t).
  \end{align*}
  Thus
  \begin{equation*}
    k(t) = \frac{G(t) - G_L(t)}{1-G_L(t)} \leq \alpha_0
  \end{equation*}
  for all $t$. Thus $\alpha_0' = \lim_{t\uparrow t^*} k(t) \leq \alpha_0$.

  \noindent
  \textbf{Show $\alpha'_0 \geq \alpha_0$:} Consider any $\gamma < \alpha_0$. We show $\gamma < \alpha_0'$. Since $\gamma < \alpha_0$,
  \begin{equation*}
    \frac{G - (1-\gamma)G_L}{\gamma}
  \end{equation*}
  is not a c.d.f. Thus there exists $t_1 < t_2$ such that
  \begin{equation} \label{eq:step1}
    \frac{G(t_1) - (1-\gamma)G_L(t_1)}{\gamma} >     \frac{G(t_2) - (1-\gamma)G_L(t_2)}{\gamma}.
  \end{equation}
  
  \underline{Case 1: L.H.S. of Equation \eqref{eq:step1} $> 1$:}
  If $G_L(t_1) = G(t_1)$, L.H.S. $=G_L(t_1) \leq 1$. Thus $G_L(t_1) \neq G(t_1)$. By Lemma \ref{lem:monotone} $G_L(t_1) \leq G(t_1) \leq 1$. Thus $G_L(t_1) < 1$. By assumption
  \begin{equation*}
    \frac{G(t_1) - (1-\gamma)G_L(t_1)}{\gamma} >  1
  \end{equation*}
  Rearranging terms
  \begin{equation}
  \label{eq:almost}
    \gamma < \frac{G(t_1) - G_L(t_1)}{1-G_L(t_1)}
  \end{equation}
  Since $G_L(t^*)=1$ and $G_L(t) < 1$, $t_1 < t^*$. Thus by Lemma \ref{lem:ktinc} the R.H.S. of Equation \eqref{eq:almost} is bounded by $\alpha_0'$.
  
  \underline{Case 2: L.H.S. of Equation \eqref{eq:step1} $\leq 1$:}
  If $t_2 \geq t^*$, then $G_L(t_2) = 1$. Since $G(t) \geq G_L(t)$ (Lemma \ref{lem:monotone}), $G(t_2)=1$. Thus the R.H.S. of Equation \eqref{eq:step1} equals 1. This violates the assumption of Case 2, thus $t_2 < t^*$.
  
  From Equation \eqref{eq:step1} we have
  \begin{equation*}
    G(t_1) - G(t_2) > (1-\gamma) (G_L(t_1) - G_L(t_2))
  \end{equation*}
which implies (since $G_L(t_1) - G_L(t_2) < 0$) that
\begin{equation} \label{eq:continue}
    \frac{G(t_2) - G(t_1)}{G_L(t_2) - G_L(t_1)} < (1-\gamma).
  \end{equation}
  From Lemma \ref{lem:ratio} we have
  \begin{equation*}
    \frac{1 - G_L(t_2)}{1-G(t_2)} = \frac{G_L(1) - G_L(t_2)}{G(1)-G(t_2)} \geq \frac{G_L(t_2) - G_L(t_1)}{G(t_2)-G(t_1)}
  \end{equation*}
  Combining this result with Equation \eqref{eq:continue} we obtain
  \begin{equation*}
    \frac{1 - G(t_2)}{1-G_L(t_2)} \leq 1-\gamma
  \end{equation*}
  which implies
  \begin{equation*}
    \gamma \leq \frac{G(t_2) - G_L(t_2)}{1-G_L(t_2)} = k(t_2)
  \end{equation*}
  Since $k(t) \uparrow$ as $t \uparrow t^*$ (see Lemma \ref{lem:ktinc}), we have the result.
\end{proof}

\begin{Lem} \label{lem:ktinc}
  $k(t)$ is increasing on $t\in [0,t^*)$.
\end{Lem}
\begin{proof}
Recall $Q(p) = \inf\{t \in (0,1] : G_L(t) \geq p\}$ and $t^* = Q(1)$. Note that with $a,b,c,d > 0$ and $a/b < c/d$,
\begin{equation*}
  \frac{a + c}{b + d} > \frac{a}{b}.
\end{equation*}
Next note that by Lemma \ref{lem:ratio}, for $t^* > t_2 > t_1$,
\begin{equation*}
\frac{G(t_2) - G(t_1)}{G_L(t_2) - G_L(t_1)} >   \frac{1 - G(t_2)}{1-G_L(t_2)}.
\end{equation*}
Thus we have
\begin{align*}
  1-k(t_1) &= \frac{1-G(t_1)}{1-G_L(t_1)}\\
  &= \frac{1-G(t_2) + G(t_2) - G(t_1)}{1-G_L(t_2) + G_L(t_2) -G_L(t_1)}\\
  &\geq \frac{1-G(t_2)}{1-G_L(t_2)}\\
  &=1-k(t_2).
\end{align*}
\end{proof}

\begin{Lem} \label{lem:monotone}
\begin{equation*}
\frac{g_L(t)}{g(t)} = \frac{1-\pi}{\pi}\frac{t}{1-t}
\end{equation*}
and 
\begin{equation*}
G(t) \geq G_L(t)
\end{equation*}
for all $t$.
\end{Lem}
\begin{proof}
Define $A = \{x : \frac{\pi f_L(x)}{\pi f_L(x) + (1-\pi)f(x)}=t\}$
\begin{align*}
\frac{g_L(t)}{g(t)} &= \frac{\int_A f_L(x)}{\int_A f(x)}\\
&= \frac{1-\pi}{\pi} \frac{\int_A \frac{\pi f_L(x)}{\pi f_L(x) + (1-\pi)f(x)}\pi f_L(x) + (1-\pi)f(x) }{\int_A \frac{(1-\pi) f(x)}{\pi f_L(x) + (1-\pi)f(x)} \pi f_L(x) + (1-\pi)f(x)}\\
&= \frac{1-\pi}{\pi} \frac{\int_A t(\pi f_L(x) + (1-\pi)f(x))}{\int_A (1-t)(\pi f_L(x) + (1-\pi)f(x))}\\
&= \frac{1-\pi}{\pi} \frac{t}{1-t}
\end{align*}
Thus $\frac{g_L(t)}{g(t)} \, \, \uparrow$ as $t \, \, \uparrow$. 
Since $\frac{g_L(t)}{g(t)}$ is monotone increasing in $t$, $g_L$ stochastically dominates $g$ and thus $G(t) \geq G_L(t)$ for all $t$. Formally this can be shown by considering any $t_2 > t_1$ and noting
\begin{equation*}
\frac{g_L(t_2)}{g(t_2)} \geq \frac{g_L(t_1)}{g(t_1)}.
\end{equation*}
Thus
\begin{equation}
\label{eq:mlr_imp}
g_L(t_2)g(t_1) \geq g_L(t_1)g(t_2).
\end{equation}
Integrating \eqref{eq:mlr_imp} from 0 to $t_2$ with respect to $t_1$ we obtain
\begin{equation*}
g_L(t_2)G(t_2) \geq G_L(t_2)g(t_2),
\end{equation*}
which implies
\begin{equation}
\label{eq:p1}
\frac{G(t)}{G_L(t)} \geq \frac{g(t)}{g_L(t)}.
\end{equation}
Integrating \eqref{eq:mlr_imp} from $t_1$ to $1$ with respect to $t_2$ we obtain
\begin{equation*}
(1-G_L(t_1))g(t_1) \geq g_L(t_1)(1-G(t_1)),
\end{equation*}
which implies
\begin{equation}
\label{eq:p2}
\frac{g(t)}{g_L(t)} \geq \frac{1-G(t)}{1-G_L(t)}.
\end{equation}
Combining Equations \eqref{eq:p1} and \eqref{eq:p2} we have
\begin{equation*}
\frac{G(t)}{G_L(t)} \geq \frac{1-G(t)}{1-G_L(t)}
\end{equation*}
which implies the result
\begin{equation*}
G(t) \geq G_L(t).
\end{equation*}
\end{proof}

\begin{Lem} [Ratio] \label{lem:ratio}
  For all $0 \leq t_1 < t_2 \leq 1$ where $G(t_2) - G(t_1) > 0$ we have
\begin{equation*}
  \frac{1-\pi}{\pi} \frac{t_1}{1-t_1} < \frac{G_L(t_2) - G_L(t_1)}{G(t_2) - G(t_1)} \leq \frac{1-\pi}{\pi} \frac{t_2}{1-t_2}
\end{equation*}
where $1/0 \equiv \infty$.
\end{Lem}
\begin{proof}
  The classifier is
  \begin{equation*}
    C(x) = \frac{\pi f_L(x)}{\pi f_L(x) + (1-\pi)f(x)} = \frac{1}{1 + \frac{1-\pi}{\pi}\frac{f(x)}{f_L(x)}}
  \end{equation*}
  Define $A_t = \{x : C(x) \leq t\} = \{x : \frac{1-t}{t}\frac{\pi}{1-\pi}f_L(x) \leq f(x)\}$. Therefore on the set $A_{t_2} \cap A_{t_1}^C$ we have
  \begin{equation*}
    \frac{1-t_2}{t_2}\frac{\pi}{1-\pi}f_L(x) \leq f(x) < \frac{1-t_1}{t_1}\frac{\pi}{1-\pi}f_L(x)
  \end{equation*}
  So
\begin{align*}
  \frac{G_L(t_2) - G_L(t_1)}{G(t_2) - G(t_1)} &= \frac{\int_{A_{t_2}\cap A_{t_1}^C} f_L(x)}{\int_{A_{t_2}\cap A_{t_1}^C} f(x)}\\
  &> \frac{\int_{A_{t_2}\cap A_{t_1}^C} f_L(x)}{\frac{1-t_1}{t_1}\frac{\pi}{1-\pi}\int_{A_{t_2}\cap A_{t_1}^C} f_L(x)}\\
  &=\frac{t_1}{1-t_1}\frac{1-\pi}{\pi}.
\end{align*}
We can obtain the upper bound in an identical manner.
\end{proof}

\begin{Lem} \label{lem:CVM_rate}
\begin{align*}
n^{\beta-\eta} d_n(G, G_n) &= o_P(1), \\
n^{\beta-\eta} d_n(G_L, G_{L,n}) &= o_P(1).
\end{align*} 
\end{Lem}

\begin{proof}
\begin{align*}
n^{\beta-\eta} d_n(G, G_n) &= \sqrt{\int \left[\underbrace{n^{-\eta}}_{=o_P(1)}\underbrace{n^{\beta}\left(G_n(t) - G(t) \right)}_{=O_P(1)}\right]^2d G_n(t)},
\end{align*}
where $n^{\beta}\left(G_n(t) - G(t) \right) = O_P(1)$ uniformly, and then $n^{-\eta}n^{\beta}\left(G_n(t) - G(t) \right) = o_P(1)$ uniformly. Therefore
\begin{align*}
n^{\beta-\eta} d_n(G, G_n) &\leq \sup_t|n^{-\eta} n^{\beta}\left(G_n(t) - G(t) \right)| = o_P(1).
\end{align*}
The $G_L$, $G_{L,n}$ case can be proven in an identical manner.
\end{proof}
% In typical Cram\'{e}r-von Mises statistic, the rate is usually $n^{1/2}$. But now we could have more than one dimensions, so rate $n^\beta$ here is expected to slower than the usual one.
 
\begin{Lem} \label{lem:PS_distance}
For $1 \geq \gamma \geq \alpha_0$, 
\begin{align*}
\gamma d_n(\widehat{G}_{s,n}^\gamma,\check{G}_{s,n}^\gamma) \leq d_n(G, G_n) + (1-\gamma)d_n(G_L, G_{L,n}).
\end{align*}
Thus,
\begin{align*}
\gamma d_n(\widehat{G}_{s,n}^\gamma,\check{G}_{s,n}^\gamma) \rightarrow \begin{cases}
0 & \text{ if } \gamma \geq \alpha_0,\\ 
>0 & \text{ if } \gamma < \alpha_0 .
\end{cases}
\end{align*}
\end{Lem}

\begin{proof}
Let 
\begin{align*}
G_s^\gamma = \frac{G-(1-\gamma)G_L}{\gamma}.
\end{align*}
If $\gamma \geq \alpha_0$, then 
\begin{align*}
\gamma d_n(\widehat{G}_{s,n}^\gamma,\check{G}_{s,n}^\gamma) &\leq \gamma d_n(\widehat{G}_{s,n}^\gamma,G_{s}^\gamma) \\
&\leq d_n(G, G_n) + (1-\gamma)d_n(G_L, G_{L,n}).
\end{align*}
The first inequality holds by the definition of $\check{G}_{s,n}^\gamma$ due to the fact that $G_s^\gamma$ is a valid CDF when $1\geq \gamma \geq \alpha_0$, and the second inequality is due to triangle inequality.

Now we prove the limit property of $\gamma d_n(\widehat{G}_{s,n}^\gamma,\check{G}_{s,n}^\gamma)$. If $\gamma \geq \alpha_0$, then $\gamma d_n(\widehat{G}_{s,n}^\gamma,\check{G}_{s,n}^\gamma) \rightarrow 0$ since $d_n(G, G_n) \rightarrow 0$ and $d_n(G_L, G_{L,n}) \rightarrow 0$ by Lemma \ref{lem:CVM_rate}. If $\gamma < \alpha_0$, by the definition of $\alpha_0^G$, $G_s^\gamma$ is not a valid c.d.f.. Pointwise, $\widehat{G}_{s,n}^\gamma \rightarrow G_s^\gamma$. So for large $n$, $\widehat{G}_{s,n}^\gamma$ is not valid c.d.f., while $\check{G}_{s,n}^\gamma$ is always a c.d.f.. So $\gamma d_n(\widehat{G}_{s,n}^\gamma,\check{G}_{s,n}^\gamma) $ would converge to some positive constant.
\end{proof}

\begin{Lem} \label{lem:PS_convex}
$B_n := \{\gamma \in [0,1] : n^{\beta-\eta} \gamma d_n(\widehat{G}_{s,n}^\gamma,\check{G}_{s,n}^\gamma) \leq c_n \}$ is convex. Thus, $B_n = (\widehat{\alpha}_0^{c_n},1]$ or $B_n = [\widehat{\alpha}_0^{c_n},1]$.
\end{Lem}
\begin{proof}
Obviously, $1 \in B_n$.
Assume $\gamma_1 \leq \gamma_2$ from $B_n$, let $\gamma_3 = \xi \gamma_1 + (1-\xi)\gamma_2$, where $\xi \in [0,1]$. Then by definition of $\widehat{G}_{s,n}^\gamma$,
\begin{align*}
\xi \gamma_1 \widehat{G}_{s,n}^{\gamma_1} + (1-\xi) \gamma_2 \widehat{G}_{s,n}^{\gamma_2} = \gamma_3 \widehat{G}_{s,n}^{\gamma_3}.
\end{align*}
Note that $\frac{1}{\gamma_3}\left( \xi \gamma_1 \check{G}_{s,n}^{\gamma_1} + (1-\xi) \gamma_2 \check{G}_{s,n}^{\gamma_2} \right)$ is a valid c.d.f. We have $\gamma_3 \in B_n$ because
\begin{align*}
&d_n(\widehat{G}_{s,n}^{\gamma_3},\check{G}_{s,n}^{\gamma_3}) \\
& \leq d_n\left( \widehat{G}_{s,n}^{\gamma_3}, \frac{1}{\gamma_3}\left( \xi \gamma_1 \check{G}_{s,n}^{\gamma_1} + (1-\xi) \gamma_2 \check{G}_{s,n}^{\gamma_2} \right)\right) \\
&= d_n\left(\frac{1}{\gamma_3}\left( \xi \gamma_1 \widehat{G}_{s,n}^{\gamma_1} + (1-\xi) \gamma_2 \widehat{G}_{s,n}^{\gamma_2} \right), \frac{1}{\gamma_3}\left( \xi \gamma_1 \check{G}_{s,n}^{\gamma_1} + (1-\xi) \gamma_2 \check{G}_{s,n}^{\gamma_2} \right)\right) \\
&\leq \frac{\xi\gamma_1}{\gamma_3} d_n(\widehat{G}_{s,n}^{\gamma_1},\check{G}_{s,n}^{\gamma_1}) +  \frac{(1-\xi)\gamma_2}{\gamma_3} d_n(\widehat{G}_{s,n}^{\gamma_2},\check{G}_{s,n}^{\gamma_2}) \\
&\leq \frac{\xi\gamma_1}{\gamma_3} \frac{c_n}{n^{\beta-\eta}\gamma_1} +  \frac{(1-\xi)\gamma_2}{\gamma_3} \frac{c_n}{n^{\beta-\eta}\gamma_2} = \frac{c_n}{n^{\beta-\eta}\gamma_3}.
\end{align*}
\end{proof}

\end{document}